\documentclass[preprint,12pt]{elsarticle}

\usepackage{amsmath,amsfonts,amsthm,amssymb,paralist,subfigure,graphicx,amsbsy,float,epsfig,color}
\usepackage{cuted,mathtools,lipsum}

\usepackage[ruled,vlined,linesnumbered]{algorithm2e}
\usepackage{stmaryrd,url}

\usepackage{amsthm}
\newtheorem{lem}{Lemma}
\newtheorem{ass}{Assumption}

\newtheorem{rem}{Remark}

\newtheorem{cor}{Corollary}

\def\mb{\mathbf}
\def\mbb{\mathbb}
\def\mc{\mathcal}

\def\mb{\mathbf}
\def\mbb{\mathbb}
\def\mc{\mathcal}

\journal{Systems \& Control Letters}

\begin{document}
	
	\begin{frontmatter}
		
		\title{ Delay-Tolerant Augmented-Consensus-based  Distributed Directed Optimization
		}
		
		\author[Sem]{Mohammadreza Doostmohammadian}
		\affiliation[Sem]{Mechatronics Group, Faculty of Mechanical Engineering, Semnan University, Semnan, Iran, doost@semnan.ac.ir.}
		\author[nr]{Narahari Kasagatta Ramesh}
		\affiliation[nr]{School of Electrical Engineering, Aalto University, Espoo, Finland,  narahari.kasagattaramesh@aalto.fi}
		
		\author[SP]{Alireza Aghasi}
		\affiliation[SP]{Electrical Engineering and Computer Science Department, Oregon  State University, USA, alireza.aghasi@oregonstate.edu}
		
		\begin{abstract}
			Distributed optimization finds applications in large-scale machine learning, data processing and classification over multi-agent networks. In real-world scenarios, the communication network of agents may encounter latency that may affect the convergence of the optimization protocol.
			This paper addresses the case where the information exchange among the agents (computing nodes) over data-transmission channels (links) might be subject to communication time-delays, which is not well addressed in the existing literature. Our proposed algorithm improves the state-of-the-art by handling heterogeneous and arbitrary but bounded and fixed (time-invariant) delays over general strongly-connected directed networks. Arguments from matrix theory, algebraic graph theory, and augmented consensus formulation are applied to prove the convergence to the optimal value. Simulations are provided to verify the results and compare the performance with some existing delay-free algorithms.
		\end{abstract}
		
		%%Graphical abstract
%		\begin{graphicalabstract}
%			\includegraphics{grabs}
%		\end{graphicalabstract}
		
		%%Research highlights
		\begin{highlights}
			\item Introducing a new distributed optimization technique based on algebraic graph theory and augmented consensus protocols.
			\item Handling heterogeneous, arbitrary, time-invariant but bounded delays over general strongly-connected directed networks
		\end{highlights}
		

	\end{frontmatter}
	
	\section{Introduction} \label{sec_intro}
	In recent years, the study of distributed (or decentralized) algorithms for optimization, learning, and classification over a network of computing nodes/agents has gained significant attention due to adavances in cloud-computing and parallel data processing \cite{ageed2024distributed}. These networks consist of multiple agents, each with limited computational and communication capabilities, working collaboratively to solve optimization problems or learn from data in a distributed manner. However, a critical challenge in these networks is the presence of time-delays \cite{zhang2022time,DOOSTMOHAMMADIAN2025106062}, which can arise from communication latencies, processing times, or network congestion. Time-delays can severely impact the performance and convergence of distributed algorithms, making it essential to develop robust methods that can handle such delays. This paper explores the theoretical foundations and practical implementations of distributed optimization and learning algorithms that are resilient to time-delays, providing mathematical proofs, analysis, and potential applications.
	
	\subsection{Problem and Contributions}
	%\subsection{Distributed Optimization with No Time-Delay}
	In distributed optimization, the idea is to optimize a cost function (or loss function) over a network of computing nodes. The objective function is the sum of some local cost functions at the nodes, and the goal is to optimize this objective using locally defined gradient-based algorithms. The common form of the optimization problem is, 
	\begin{align}
		\min_{\mb{z}}
		F(\mb z) = \sum_{i=1}^{N} f_i(\mb{z})
		\label{eq_prob}
	\end{align}
	with state parameter $\mb{z} \in \mathbb{R}^{m}$. Functions $f_i:\mathbb{R}^m \mapsto \mathbb{R}$ are strongly convex, differentiable with Lipschitz gradients, and denote the objective function (cost, loss, etc.) at computing node $i$. It is assumed that the optimal point $\underline{\mb{z}}^* = \min_{\mb{z}} F(\mb z)$ for this problem exists.   The primary work \cite{nedic2009distributed} introduces subgradient algorithms to solve this problem. 
	ADD-OPT algorithm \cite{addopt} and its recent stochastic version S-ADD-OPT \cite{qureshi2020s} are popular algorithms to solve problem~\eqref{eq_prob}. These algorithms work over strongly-connected directed networks with irreducible column stochastic adjacency matrices, and are granted with (i) constant step-size in contrast to existing diminishing step-size algorithms, (ii) providing accelerated convergence by tuning the step-size over a wide range, and (iii) linear convergence rate for strongly convex cost functions. Other existing distributed algorithms include: event-triggered-based second-order multi-agent systems \cite{xu2023fully,CAI2023105619}, double step-size solutions for nonsmooth optimization \cite{yi2021distributed},  reduced-complexity and flexible algorithms \cite{liang2024hierarchically}, primal-dual subgradient-based solutions \cite{zhu2023primal}, EXTRA algorithm for first-order consensus-based optimization \cite{shi2015extra}, push-pull gradient-based methods \cite{nedic2014distributed}, and the solutions based on alternating direction method of multipliers (ADMM) \cite{zarepisheh2018computation,chang2014multi,song2016fast,LU2024105698}. The literature also includes distributed constrained optimization with application to resource allocation under time-delay. For, example, DTAC-ADMM discusses ADMM-based distributed resource allocation under time-delay \cite{cdc22}. Similarly, \textit{asynchronous} ADMM-based resource allocation algorithms are proposed in \cite{jiang2022distributed}. These works consider distributed optimization subject to a coupling resource-demand balance constraint, where the objective functions are decoupled and local. \textit{Asynchronous distributed optimization} is also discussed in \cite{assran2020advances,wu2025asynchronous}, where  agents perform local computations and communications without requiring global synchronization. In such methods, each node updates its local model using the most recently available information from neighbors (which may be received at irregular times). Recall that \textit{scalability} is a key advantage of existing distributed optimization techniques, which follows the polynomial-order complexity of the algorithms. Polynomial-order complexity ensures computationally-efficient solutions as the number of agents or decision variables increases, making it feasible to deploy these algorithms on large-scale networks.

	%\subsection{Contributions and Literature Review}
	In this work, as the main contribution, we extend such distributed optimization algorithms to further address arbitrary and bounded time-delays over multi-agent networks. Latency is primarily addressed in consensus literature including: resilient consensus with $l$-hop communication \cite{shang2023resilient}, multi-agent consensus subject to uncertainties and time-varying delays \cite{shang2014average},  group consensus over digraphs subject to noise and latency \cite{shang2015group}, continuous-time linear average consensus with constant delays at all nodes \cite{olfati2004consensus}, discrete-time consensus algorithms with constant communication delays \cite{seuret2008consensus}, discrete-time consensus over digraphs under heterogeneous time-delays \cite{hadjicostis2013average}. These works are advantageous as they provide rigorous stability/convergence analysis applicable to other distributed setups; however, they mostly assume constant homogeneous delays. For a review of consensus algorithms under time-delays and their advantages/disadvantages, refer to \cite{behjat2024modeling}.  The concept of time-delay is not sufficiently addressed in distributed optimization literature. The inherent time-delay of information exchange among communicating nodes may lead the distributed optimization algorithm to lose convergence.
	The delays are typically assumed to be bounded, implying that the information sent over every link eventually reaches the destination node, i.e., no packet loss over the network. In this paper, we propose \textit{augmented} consensus-based algorithms to analyze the effect of time-delays while keeping the consensus matrix on the link weights column stochastic. Our solution can tolerate \textit{heterogeneous} communication delays
	at different links. In this regard, similar to \cite{6426375}, this work improves the existing algorithms over non-delayed networks \cite{xin9084360,csl2021,saadatniaki2020decentralized,shi2015extra,addopt,qureshi2020s,nedic2014distributed} to
	more advanced delay-tolerant solutions which are not well-addressed in the literature (to our best knowledge). This work also advances the existing ADMM-based solutions \cite{zarepisheh2018computation,chang2014multi,song2016fast} to withstand latency and network time-delays. Our proposed delay-tolerant augmented consensus-based DTAC-ADDOPT algorithm is in single time-scale,
	i.e., it performs only one step of (augmented) consensus on
	received information per iteration/epoch. This is computationally
	more efficient in contrast to the double time-scale methods \cite{jiang2021distributed,rokade2020distributed} with many steps of inner-loop consensus per iteration/epoch. Heterogeneous time-delays are considered primarily for \textit{ADMM-based} \cite{cdc22} and gradient-descent-based \cite{DOOSTMOHAMMADIAN2025106062}  \textit{equality-constraint} distributed optimization and resource allocation, but this current paper is our first paper addressing it over unconstrained ADD-OPT.

	\subsection{Applications}
	\textbf{Distributed Training for Binary Classification:}
	Consider a group of agents to classify~$N$ data points~${\boldsymbol{\chi}_i \in \mathbb{R}^{m-1}}$, ${i=1,\ldots,N}$, labeled by~${l_i \in \{-1,1\}}$. The problem is to find the partitioning hyperplane~${\boldsymbol{\omega}^\top \boldsymbol{\chi} - \nu =0}$,  for~${\boldsymbol{\chi}\in\mbb R^{m-1}}$. In the linearly non-separable case, a proper nonlinear mapping~$\phi(\cdot)$
	with \textit{kernel}~$K(\boldsymbol{\chi}_i,\boldsymbol{\chi}_j)=\phi(\boldsymbol{\chi}_i)^\top \phi(\boldsymbol{\chi}_j)$  can be found such that~${g(\widehat{\boldsymbol{\chi}})= \text{sgn}(\boldsymbol{\omega}^\top \phi(\widehat{\boldsymbol{\chi}}) - \nu)}$ determines the class of~$\widehat{\boldsymbol{\chi}}$. Agents collaboratively solve the problem by finding the optimal~$\boldsymbol{\omega}~$ and~$\nu$ and optimize the following  convex loss~\cite{chapelle2007training}:
	\begin{equation} \label{eq_svm_cent}
		\begin{aligned}
			\displaystyle
			f_i(\boldsymbol{\omega},\nu) = \boldsymbol{\omega}^\top \boldsymbol{\omega} + C \sum_{j=1}^{N} \max\{x,0\}^p
		\end{aligned}
	\end{equation}
	with~$p \in \mathbb{N}$ as the smoothness factor (which is typically a finite number),~$C>0$ as the margin size parameter, and~${x=1-l_j( \boldsymbol{\omega}^\top \phi(\boldsymbol{\chi}^i_j)-\nu)}$. The differentiable smooth equivalent of $f_i$ in Eq.~\eqref{eq_svm_cent} is in the following form (assuming large enough~$\mu>0$):
	\begin{equation} \label{eq_svm_smooth}
		f_i(\boldsymbol{\omega},\nu)=\boldsymbol{\omega}^\top \boldsymbol{\omega} + C \sum_{j=1}^{N_i} \tfrac{1}{\mu}\log (1+\exp(\mu x)).
	\end{equation}
	This problem is also known as distributed support-vector-machine (D-SVM)
	\cite{csl2021,saadatniaki2020decentralized}.
	
	\textbf{Distributed Least Squares:} In this problem, the idea is to solve the least square problem $H\mb{z}=\mb{b}$ in a distributed manner. Every agent/node $i$ takes measurement $\mb{b}_i \in \mathbb{R}^p$ and has a $p$-by-$n$ measurement matrix $H_i$ and collaboratively optimizes the private loss function in the following form
	\cite{saadatniaki2020decentralized}:
	\begin{align}
		f_i(\mb{x}) = \frac{1}{2} \|H_i \mb{z}-\mb{b}_i\|_2^2
		\label{eq_prob_ls}
	\end{align}
	This can be addressed further in the context of distributed filtering \cite{10989575,dimakis2010gossip,kar2008distributed}.
	
	\textbf{Distributed Logistic Regression:}
	In this problem each agent $i$ with access to $m_i$ training data points defined by $(c_{ij},y_{ij}) \in \mathbb{R}^p \times \{-1,1\}$, where the parameter $c_{ij}$ has $p$ features of the $j$th training data and $y_{ij}$ denotes the binary label $\{-1,+1\}$. Each agent, collaborating with others, solves and optimizes the private loss function in the following form \cite{xin9084360}:
	\begin{align}
		f_i(\mb{w},b) = \sum_{j=1}^{m_i} \log(1+\exp(-(\mb{w}^\top c_{ij}+b)y_{ij})) + \frac{\lambda}{2} \|\mb{w}\|_2^2
		\label{eq_prob_lr}
	\end{align}
	where the last term is for regularization to avoid overfitting.
	
	\subsection{Paper Organization}
	Section~\ref{sec_pre} provides the preliminary notions. Section~\ref{sec_main} gives the main DTAC-ADDOPT algorithm with proof of convergence in Section~\ref{sec_conv}. Section \ref{sec_sim} provides the simulation results on both the academic setup and the real dataset. Section \ref{sec_con} provides the concluding remarks.
	
		\subsection{Notations}
		Table~\ref{tab_notation} summarizes the notations in this paper.
		\begin{table}
			\caption{Description of notations and symbols}
			\setlength{\tabcolsep}{0.7\tabcolsep}
			\centering \label{tab_notation}
			\begin{tabular}{ *{2}{c} } 
				\hline
				\hline
				\textbf{Symbol} & \textbf{Description} \\
				\hline
				$\mc{G}$ & multi-agent network   \\ 
				$C$ & adjacency matrix of the network \\
				$\mc{N}_i$ & neighbors of agent $i$  \\
				$\mb{z}$ & global state variable  \\
				$\underline{\mb{z}}^*$ & optimal state  \\
				$F$ & global objective function      \\
				$f_i$ & local objective function at node $i$\\
				$m$ & dimention of state variable\\
				$n$ & number of nodes/agents\\
				$\tau_{ij}$ & time-delay at link $(i,j)$\\
				$\overline{\tau}$ & bound on the time-delays\\
				$\overline{C}$ & augmented adjacency matrix \\
				$\mb{y},\mb{x},\mb{g}$ & auxiliary optimization variables  \\
				$\nabla f_k$ &  gradient of function $f$ at time $k$    \\
				$\overline{\nabla f}_k$ &  gradient vector over last $\overline{\tau}$ steps at time $k$   \\
				$\alpha$ &  gradient-tracking step rate  \\
				$\widehat{\mb{z}}$ & augmented state variable  \\
				$\widehat{\mb{y}},\widehat{\mb{x}},\widehat{\mb{g}}$ & augmented auxiliary  variables  \\
				$\rho$ & spectral radius   \\
				$\mb{1}_n$ & all ones column vector of size $n$ \\
				$I_n,0_n$ & identity and zero matrix of size $n$   \\ 
				$k$ & time index  \\
				$\|\cdot\|$ & 2-norm operator \\
				$\otimes$ & Kronecker product operator  \\      
				\hline \hline
			\end{tabular}
		\end{table}

\section{Preliminaries} \label{sec_pre}
\subsection{Algebraic Graph Theory}
We consider the network of agents as a digraph (directed graph) of nodes denoted by $\mc{G}=\{\mc{V},\mc{E}\}$ with $\mc{V}$ and $\mc{E}$ respectively as the node set and link set. A link $(i,j) \in \mc{E}$ from node $i$ to node $j$ implies a communication link for message passing from agent $i$ to agent $j$. The adjacency matrix of $\mc{G}$ is denoted by $C$ where  $C_{ij}$ is the weight on the link $(j,i)$ (or $j \rightarrow i$). Define the in-neighborhood of every node $i$ as $\mc{N}_i = \{j|(j,i) \in \mc{E}\}$ or $\mc{N}_i = \{j|C_{ij} \neq 0 \}$. 
\begin{ass}
	The digraph (or network) $\mc{G}$ is strongly connected and its adjacency matrix $C$ is irreducible \cite{godsil}. Moreover, the matrix $C$ is column stochastic, i.e., $\sum_{i=1}^n C_{ij} =1$.
\end{ass}

Note that, in most directed network implementations,  agents already know their outgoing neighbor set for column-stochastic design of matrices, and the out-degree is locally available.

\subsection{Augmented Formulation}
The delay model is similar to the consensus literature \cite{hadjicostis2013average} and is clearly defined in the following assumption.
\begin{ass}
	The time-delays are considered heterogeneous (at
		different links), bounded, arbitrary, and time-invariant. An integer value $0\leq \tau_{ij}\leq \overline{\tau}$ represents the delay at link $(i,j)$. The bound $\overline{\tau}$ ensures no information loss over the
		network.
\end{ass}
We justify the above assumption. Note that, in practice, delays may change on a time-scale much slower than the algorithm step-size (or are upper-bounded by the same constant); therefore, the derived bounds using the maximum delay remain valid in practical cases. Also, in many networks, communication paths and routing remain stable for long periods. Communication latencies in these settings are dominated by propagation and queuing delays that are (on the algorithm time-scale) nearly constant and hence well modeled as time-invariant. Moreover, treating delays as fixed (but heterogeneous) provides a conservative worst-case analysis useful for algorithm design and safety guarantees.

For every set of connected nodes $(i,j)$ and $(i,k)$,
the communication delay implies the
heterogeneous scenario. Define the augmented state vectors $\widehat{\mb{z}}_k = (\mb{z}_k; \mb{z}_{k-1}; \dots; \mb{z}_{k-\overline{\tau}})$ as the column-concatenation  of delayed state vectors (";" denotes column concatenation).
Given the column stochastic consensus matrix $C$ and maximum delay $\overline{\tau}$, its \textit{augmented matrix} is defined as,
\begin{align} \label{eq_aug_C} \small
	\overline{C} = \left(
	\begin{array}{cccccc}
		C_0 & I_n & 0_n & \hdots & 0_n & 0_n \\
		C_1 &   0_n & I_n &\hdots  & 0_n& 0_n\\
		\vdots & \vdots & \vdots & \ddots & \vdots & \vdots \\
		C_{\overline{\tau}-1}  & 0_n & 0_n &  \hdots  & I_n & 0_n  \\
		C_{\overline{\tau}} & 0_n & 0_n & \hdots & 0_n & I_n
	\end{array} 
	\right), \normalsize
\end{align}
with $I_{n}$ and $0_{n}$ respectively as $n$-by-$n$ identity and zero matrices.
The non-negative matrices $C_r$ with $r \in \{0,\dots,\overline{\tau}\}$ are defined based on delay $0 \leq r \leq \overline{\tau} $ as,
\begin{align}
	C_{r,ij} = \left\{
	\begin{array}{ll}
		C_{ij}, & \text{If}~ \tau_{ij}=r  \\
		0, & \text{Otherwise}.
	\end{array}\right.
\end{align}
Assuming time-invariant delays, for every $(i,j) \in \mc{E}$, \textit{only one of the entries $C_{0,ij},C_{1,ij}, \hdots, C_{\overline{\tau},ij}$ is  equal to  $C_{ij}$  and the rest are zero}. This implies that the column-sum of the first $n$ columns of $\overline{C}$ and $C$ are equal. Note that, given a column-stochastic $C$ matrix, the augmented matrix $\overline{C}$ is also column stochastic from the definition. It should be noted that this large augmented matrix is only used in proof analysis of the proposed algorithm, and it is not practically used in the agents' dynamics (see the iterative dynamics \eqref{eq_main_mat_x1}-\eqref{eq_main_mat_g1} in the next section). 

\section{The Main Algorithm} \label{sec_main}
The main algorithm in compact matrix form (subject to no time-delay) is given as follows:
\begin{align} \label{eq_nodelay_x}
	\mb{x}_{k+1} &= {C}_k \mb{x}_{k} - \alpha \mb{g}_k \\ \label{eq_nodelay_y}
	\mb{y}_{k+1} &= {C}_k \mb{y}_{k} \\ \label{eq_nodelay_z}
	\mb{z}_{k+1} &= \frac{\mb{x}_{k+1}}{\mb{y}_{k+1}} \\ \label{eq_nodelay_g}
	\mb{g}_{k+1} &= {C}_k \mb{g}_{k} + \nabla \mb{f}_{k+1} - \nabla \mb{f}_{k}
\end{align}
For delayed case, define the augmented vectors $\widehat{\mb{x}}_k,\widehat{\mb{y}}_k,\widehat{\mb{g}}_k$ of size $n(\overline{\tau}+1)$.
%\begin{align}
%    \widehat{\mb{x}}_k = (\mb{x}_k; \mb{x}_{k-1}; \dots; \mb{x}_{k-\overline{\tau}}) \\
%    \widehat{\mb{y}}_k = (\mb{y}_k; \mb{y}_{k-1}; \dots; \mb{y}_{k-\overline{\tau}}) \\
%    \widehat{\mb{g}}_k = (\mb{g}_k; \mb{g}_{k-1}; \dots; \mb{g}_{k-\overline{\tau}})
%\end{align}
Let,
$
\mb{Y}_k = \mbox{diag}(\widehat{\mb{y}}_k).
$
Further, define the  auxiliary matrix $\Xi^{n}_{i,\overline{\tau}}$ is an $n \times (\overline{\tau}+1)n $ matrix defined as
$	\Xi^{n}_{i,\overline{\tau}}= (\mb{b}^{\overline{\tau}+1}_i \otimes I_{n})^\top $
with $\mb{b}^{\overline{\tau}+1}_i$ as the unit column-vector of the $i$'th coordinate ($1\leq i \leq {\overline{\tau}+1}$), i.e.,
$
\mb{b}^{\overline{\tau}+1}_i =( \underbrace{\overbrace{0;\dots;0}^{i-1};1;0;\dots;0}_{\overline{\tau}+1})
$
In case $\mb{x} \in \mathbb{R}^{np}$ then $\Xi^{np}_{i,\overline{\tau}} = \Xi^{n}_{i,\overline{\tau}} \otimes I_p$.
Then, putting $i=1$,
we have $\mb{x}_k = \Xi^{np}_{1,\overline{\tau}} \widehat{\mb{x}}_k, \mb{y}_k = \Xi^{np}_{1,\overline{\tau}} \widehat{\mb{y}}_k, \mb{g}_k = \Xi^{np}_{1,\overline{\tau}} \widehat{\mb{g}}_k$.
In fact, $\Xi^{np}_{1,\overline{\tau}}$ returns the first $np$ rows of the augmented vector of size $np(\overline{\tau}+1)$.

The main distributed optimization dynamics under communication time-delays are in the following vector form,
\begin{align} \label{eq_main_mat_x1}
	\mb{x}_{k+1,i} &= \sum_{j \in \mc{N}_i} \sum_{r=0}^{\overline{\tau} } C_{k,ij} \mc{I}_{k-r,ij}(r) \mb{x}_{k-r,j}  - \alpha  \mb{g}_{k,i} \\ \label{eq_main_mat_y1}
	\mb{y}_{k+1,i} &=  \sum_{j \in \mc{N}_i} \sum_{r=0}^{\overline{\tau} } C_{k,ij} \mc{I}_{k-r,ij}(r) \mb{y}_{k-r,j}\\ \label{eq_main_mat_z1}
	\mb{z}_{k+1,i} &=   \frac{\mb{x}_{k+1,i}}{\mb{y}_{k+1,j}} \\ \label{eq_main_mat_g1}
	\mb{g}_{k+1,i} &= \sum_{j \in \mc{N}_i} \sum_{r=0}^{\overline{\tau} } C_{k,ij} \mc{I}_{k-r,ij}(r) \mb{g}_{k-r,j} +  (\nabla \mb{f}_{k+1,i} - \nabla \mb{f}_{k,i})
\end{align}
where $\nabla \mb{f}_{k+1,i}$ denotes $\nabla \mb{f}_{i}(\mb{z}_{k+1,i})$ and $\mc{I}$ is the indicator function,
\begin{align} \label{eq_mcI}
	\mc{I}_{k,ij}(\tau) = \left\{ \begin{array}{ll}
		1, & \text{if}~  \tau_{ij}(k) = \tau,\\
		0, & \text{otherwise}.
	\end{array}\right.
\end{align}
 In practice, agents use Eqs.~\eqref{eq_main_mat_x1}-\eqref{eq_main_mat_g1} to update their states, i.e., the recipient agents use the neighboring data as they arrive with some possible delays.
The proposed solution is summarized in Algorithm \ref{alg_1}.
\begin{algorithm}
	%	\KwResult{Final state $\mb{x}(k)$ and cost $F(\mb{x}(k))$}
	\textbf{Input:}  $W$, $\mc{G}$, $\alpha$, $\overline{\tau}$, $\mb{f}_i(\cdot)$\;
	\textbf{Initialization:} set $k=0$, node $i$ sets $\mb{y}_{0,i}=1$, $\mb{g}_{0,i}=\nabla \mb{f}_{0,i}$ and randomly sets $\mb{x}_{0,i}$ \;
	\While{termination criteria NOT true}{
		Each node $i$ receives a possibly delayed packet from $j \in \mc{N}^-_i$ and computes Eq.~\eqref{eq_main_mat_x1}-\eqref{eq_main_mat_g1}\;
		Each node $i$ shares $\mb{x}_{k+1,i},\mb{y}_{k+1,i},\mb{g}_{k+1,i}$ to neighbors $j \in \mc{N}^+_i$ \;
		Sets $k \leftarrow k+1$\;
	}
	\textbf{Return} Final state $\mb{z}_{k+1,i}$ and cost $\mb{f}_i(\mb{z}_{k+1,i})$ \;	
	\caption{DTAC-ADDOPT}
	\label{alg_1}
\end{algorithm}

Equivalently in the matrix form,
%\begin{align} \label{eq_main_mat_x}
%    \widehat{\mb{x}}_{k+1} &= \overline{C}_k \widehat{\mb{x}}_{k} - \alpha \Xi^{np}_{i,\overline{\tau}} \widehat{\mb{g}}_k \\ \label{eq_main_mat_y}
%    \widehat{\mb{y}}_{k+1} &= \overline{C}_k \widehat{\mb{y}}_{k} \\ \label{eq_main_mat_z}
%    \mb{z}_{k+1} &= \frac{\mb{x}_{k+1}}{\mb{y}_{k+1}} \\ \label{eq_main_mat_g}
%    \widehat{\mb{g}}_{k+1} &= \overline{C}_k \widehat{\mb{g}}_{k} +  \mb{b}^{\overline{\tau}+1}_1 \otimes (\nabla \mb{f}_{k+1} - \nabla \mb{f}_{k})
%\end{align}
%
%Or we can use sth like the following,
\begin{align} \label{eq_main_mat_x2}
	\widehat{\mb{x}}_{k+1} &= \overline{C}_k \widehat{\mb{x}}_{k} - \alpha  \widehat{\mb{g}}_k \\ \label{eq_main_mat_y2}
	\widehat{\mb{y}}_{k+1} &= \overline{C}_k \widehat{\mb{y}}_{k} \\ \label{eq_main_mat_z2}
	\mb{z}_{k+1} &= \Xi^{np}_{i,\overline{\tau}} \widehat{\mb{z}}_{k+1},~ \widehat{\mb{z}}_{k+1} =  \frac{\widehat{\mb{x}}_{k+1}}{\widehat{\mb{y}}_{k+1}} \\ \label{eq_main_mat_g2}
	\widehat{\mb{g}}_{k+1} &= \overline{C}_k \widehat{\mb{g}}_{k} +   %\Xi^{np}_{i,\overline{\tau}}
	(\overline{\nabla \mb{f}}_{k+1} - \overline{\nabla \mb{f}}_{k})
\end{align}
where, $ \overline{\nabla \mb{f}}_{k} := (\nabla \mb{f}_{k}; \nabla \mb{f}_{k-1}; \dots ; \nabla \mb{f}_{k-\overline{\tau}})$.
Augmented matrix $\overline{C}_k$ (defined in Section~\ref{sec_pre}) is column-stochastic \cite{hadjicostis2013average}.
%
%\begin{align}\label{eq_aug_C}
%	\overline{C}_k= \left(
%	\begin{array}{ccccc}
%		C_k^0 & I_{np} & 0_{np} & \hdots &  0_{np} \\
%		C_k^1 &   0_{np} & I_{np} &\hdots  & 0_{np}\\
%		\vdots &  \vdots & \vdots  &  \ddots  & \vdots  \\
%		C_k^{\overline{\tau}-1} & 0_{np} & 0_{np} & \hdots & I_{np} \\
%		C_k^{\overline{\tau}}  & 0_{np} & 0_{np} & \hdots & 0_{np}
%	\end{array}	
%	\right),
%\end{align}
%
%\begin{align}
%	[\overline{C}^r_{ij}] = \left\{
%	\begin{array}{ll}
%		C_{ij}, & \text{If}~ \tau_{ij}=r, (i,j) \in \mc{E}  \\
%		0, & \text{Otherwise}.
%	\end{array}\right.
%\end{align}
%Further,   matrix $\mc{C}_k$ defines the zero-nonzero structure of $\overline{C}_k$, i.e.,
%\begin{align}
%	[\mc{C}_{ij}]_k = \left\{
%	\begin{array}{ll}
%		1, & \text{If}~ [\overline{c}_{ij}^r]_k \neq 0  \\
%		0, & \text{Otherwise}.
%	\end{array}\right.
%\end{align}
For the proposed solution, one can substitute the strong-connectivity of $\mc{G}$ (irreducibility of $C_k$) in  \cite{nedic2014distributed,addopt} by the irreducibility of $\overline{C}_k$. Note that given a column-stochastic consensus matrix $C$, its augmented consensus version $\overline{C}_k$ is also column-stochastic.
\begin{lem} \label{lem_nedic}
	There exists $0<\gamma_1<1$ and $0<T<\infty$ such that
	\begin{align}
		\|\mb{Y}_k - \mb{Y}_\infty \|_2 \leq T \gamma_1^k
	\end{align}
\end{lem}
\begin{proof}
	The proof follows from \cite{nedic2014distributed,addopt,hadjicostis2013average,blondel2005convergence}.
	
	The proof of $\overline{C}_{l+k}\dots \overline{C}_{k+1}\overline{C}_{k}$ being SIA is given in \cite{hadjicostis2013average}.  The SIA property used in \cite{nedic2014distributed,addopt,blondel2005convergence} to prove the lemma in the absence of time-delays. Recall that from the definition of the spectral radius \cite{blondel2005convergence},
	\begin{align}
		\rho(\overline{C}) = \lim_{k\rightarrow \infty} \| \overline{C}_1 \overline{C}_2 \dots \overline{C}_k\|^k
	\end{align}
	Then from \cite{blondel2005convergence}, $\gamma_1 > \rho(\overline{C}) $. This value can be compared with $\gamma > \rho({C}) $ (for the non-delayed case) given in \cite{addopt}. It can be shown from \cite[Appendix]{delay_est} that $\rho(\overline{C})\leq \rho(C)^{\frac{1}{1+\overline{\tau}}}$ in Lemma~\ref{lem_rho_tau}. This implies that one can choose $\gamma_1 = \gamma^{\overline{\tau}+1} $ for example. This implies that the convergence rate may be reduced by a power ${\overline{\tau}+1}$.
\end{proof}
%   May also consider $\gamma=(1-\frac{1}{n^n})$ as in \cite{nedic2014distributed} (assuming $B=1$, i.e., assuming strong-connectivity at every time instead of uniform strong-connectivity over B time-steps), then $\gamma_1=(1-\frac{1}{n(\overline{\tau}+1)^{n(\overline{\tau}+1)}})$ (not sure about this!)

\begin{cor}
	The proof can be extended to uniformly strongly connected graphs over $B$ time-steps (or B-connected networks) as described in \cite{nedic2014distributed}. In this scenario, the multi-agent network is not necessarily connected at all times, but its union is connected over $B$ time-steps, i.e., $\cup_{t_k}^{t_k+B} \mc{G}_k$ is connected for all steps $k \geq 0$.
\end{cor}
The following lemma from our previous work \cite{delay_est} relates the spectral property of the delayed and deay-free system matrices.
\begin{lem} \label{lem_rho_tau}
	\cite{delay_est}
	Given a matrix $A$ with $\rho(A) < 1$ and its augmented form $\overline{A}$ from \eqref{eq_aug_C}, we have
	$$\rho(\overline{A}) \leq \rho(A)^{\frac{1}{1+\overline{\tau}}} < 1 $$
	Similarly, if $\rho(A) = 1$, then $\rho(\overline{A})=1$.
\end{lem}

Define,
\begin{align}
	y &:= \sup_k \| \mb{Y}_k\| \\
	y_- &:= \sup_k \| \mb{Y}^{-1}_k\|
\end{align}
and recall the column-stochasticity of $\overline{C}_k$. Then, the following lemma holds.
\begin{lem} \label{lem_sigma}
	For $\mb{a} \in  \mathbb{R}^{np(\overline{\tau}+1)}$ and $ \mb{Y}_\infty = \lim_{k \rightarrow \infty} \mb{Y}_k$ from  $\mb{Y}_k = \mbox{diag}(\widehat{\mb{y}}_k)$, there exist $0<\sigma <1$ for some $\overline{\tau}$,
	\begin{align}
		\|\overline{C}_k \mb{a} - \mb{Y}_\infty \overline{\mb{a}} \| \leq \sigma \|\mb{a} - \mb{Y}_\infty \overline{\mb{a}} \|
	\end{align}
\end{lem}
\begin{proof}
	The proof follows from the column-stochasticity of $\overline{C}$ and Lemma~\ref{lem_rho_tau}.  For any  $\mb{a} \in  \mathbb{R}^{np}$ and $\overline{\mb{a}} = \frac{1}{n} (\mb{1}_n \otimes I_{p})(\mb{1}^\top_n \otimes I_{p})\mb{a}$, there exist $0<\sigma_1<1$ \cite{addopt},
	\begin{align}
		\|{C}_k \mb{a} - \mb{Y}_\infty \overline{\mb{a}} \| \leq \sigma_1 \|\mb{a} - \mb{Y}_\infty \overline{\mb{a}} \|
	\end{align}
	Irreducible column-stochastic $C$ with positive diagonals implies $\rho(C) =1 $. Let $\pi$ satisfy $C\pi = \pi$ and $\mb{1}_n^\top \pi = 1$. $C_\infty = \lim_{k \rightarrow \infty} C^k = \pi \mb{1}_n^\top \otimes I_p$. In the presence of time delays, if $\tau_{ij} = \overline{\tau}$,$\forall i,j$, then the proof for $\overline{\pi}$ (the augmented version of ${\pi}$) similarly follows. In this case, $\overline{C}$ is irreducible,
	column-stochastic, and with proper column/row permutations, it can be transformed into a form with positive diagonals. Then, Perron-Frobenius theorem follows and $\rho(\overline{C}) =1 $ with other eigenvalue
	than $1$ strictly less than $\rho(\overline{C})$.
	Then, there exist (strictly positive) right-eigenvector $\overline{\pi}$  corresponding to the eigenvalue $1$ of $\overline{C}$ such that
	$\overline{C}_\infty = \lim_{k \rightarrow \infty} \overline{C}^k = \overline{\pi} \mb{1}_{n(\overline{\tau}+1)}^\top \otimes I_p $ (for example, $\overline{\pi} = \mb{1}_{n(\overline{\tau}+1}$) and the proof exactly follows.
	In case,  $\tau_{ij} \leq \overline{\tau}$, then $\overline{\pi}$ is not strictly positive but it is non-negative.  Following from \cite[Lemma~4]{delay_est}, $\overline{C}$ has some more zero eigenvalues.
	With $\overline{C}_\infty =  \overline{\pi} \mb{1}_{n(\overline{\tau}+1)}^\top \otimes I_p $, it follows that,
	\begin{align}
		\overline{C} \overline{C}_\infty &= \overline{C}_\infty. \\
		\overline{C}_\infty \overline{C}_\infty &= \overline{C}_\infty.
	\end{align}
	and $ \frac{1}{n(\overline{\tau}+1)} \mb{Y}_\infty (\mb{1}_{n(\overline{\tau}+1)} \otimes I_p)(\mb{1}_{n(\overline{\tau}+1)}^\top \otimes I_p) = \overline{C}_\infty$,
	\begin{align}
		\overline{C}\mb{a} - \mb{Y}_\infty \overline{\mb{a}} = (\overline{C}-\overline{C}_\infty)(\mb{a} - \mb{Y}_\infty \overline{\mb{a}}).
	\end{align}
	Next,
	\begin{align}
		\rho(\overline{C} - \overline{C}_\infty) = \rho(\overline{C} - \overline{\pi} \mb{1}_{n(\overline{\tau}+1)}^\top \otimes I_p)<1
	\end{align}
	Then,
	\begin{align}
		\|\overline{C}\mb{a} - \mb{Y}_\infty \overline{\mb{a}}\| \leq \|\overline{C}-\overline{C}_\infty\| \|\mb{a} - \mb{Y}_\infty \overline{\mb{a}}\|.
	\end{align}
	where $\sigma = \|\overline{C}-\overline{C}_\infty\|$.
	%Note that in \cite{addopt}, $\sigma_1 = \|{C}-{C}_\infty\|$, then similar to Lemma 5 in \cite{delay_est}.
\end{proof}

Finding an exact relation between $\sigma$ and $\sigma_1$ as a function of $\overline{\tau}$ could be one direction of future research.
Here, the relation between $\sigma=\left\| \overline{C} - \overline{C}_{\infty} \right\|$ (the augmented case) and $\sigma_1=\left\| C - C_{\infty} \right\|$ (the delay-free case) is approximated in the following lemma.
\begin{lem} \label{lem_rho_tau2}
	\cite{delay_est}
	Given a matrix $C$ with $\rho(C) < 1$  and its column-augmented form $\overline{C}$, we have
	$$\rho(\overline{C}) \leq \rho(C)^{\frac{1}{1+\overline{\tau}}} < 1 $$
	If $\rho(C) = 1$, then $\rho(\overline{C})=1$.
\end{lem}
This lemma implies that $\sigma \leq \sigma_1^{\frac{1}{1+\overline{\tau}}} < 1$, which says that by increasing $\overline{\tau}$, $\sigma$ gets closer to $1$.

In the absence of delays, from \cite{addopt} we have,
\begin{align}
	\underline{\overline{\mb{x}}}_k &:= \frac{1}{n} (\mb{1}_n \otimes I_{p})(\mb{1}^\top_n \otimes I_{p})\mb{x}_k \\
	\underline{\overline{\mb{g}}}_k &:= \frac{1}{n} (\mb{1}_n \otimes I_{p})(\mb{1}^\top_n \otimes I_{p})\mb{g}_k \\
	\mb{z}^* &:= \mb{1}_n \otimes \underline{\mb{z}}^* \\
	\underline{\mb{h}}_k &:= \frac{1}{n} (\mb{1}_n \otimes I_{p})(\mb{1}^\top_n \otimes I_{p}) \nabla \mb{f}_{k} \\
	\underline{\mb{q}}_k &:=  \frac{1}{n} (\mb{1}_n \otimes I_{p})(\mb{1}^\top_n \otimes I_{p}) \nabla \mb{f}(\underline{\overline{\mb{x}}}_k).
\end{align}
and, in the presence of communication time-delays (with max delay $\overline{\tau}$), the variables change to the augmented version as
\begin{align}
	\overline{\mb{x}}_k &:=  (\underline{\overline{\mb{x}}}_k;\underline{\overline{\mb{x}}}_{k-1}; \dots; \underline{\overline{\mb{x}}}_{k-\overline{\tau}}) \\
	\overline{\mb{g}}_k &:=  (\underline{\overline{\mb{g}}}_k;\underline{\overline{\mb{g}}}_{k-1}; \dots; \underline{\overline{\mb{g}}}_{k-\overline{\tau}})  \\
	\mb{z}^* &:= \mb{1}_{n(\overline{\tau}+1)} \otimes \underline{\mb{z}}^* \\
	\mb{h}_k &:=  (\underline{\mb{h}}_k;\underline{\mb{h}}_{k-1}; \dots; \underline{\mb{h}}_{k-\overline{\tau}}) \\
	\mb{q}_k &:=    (\underline{\mb{q}}_k;\underline{\mb{q}}_{k-1}; \dots; \underline{\mb{q}}_{k-\overline{\tau}}).
\end{align}
Let's define the following variables for the proof analysis:
\begin{align}
	\mb{t}_k &:= \left(
	\begin{array}{c}
		\|\widehat{\mb{x}}_k - \mb{Y}_\infty \overline{\mb{x}}_k\| \\
		\| \overline{\mb{x}}_k - \mb{z}^*\|_2 \\
		\|\widehat{\mb{g}}_k - \mb{Y}_\infty \mb{h}_k\|
	\end{array}	
	\right), \\
	\mb{s}_k &:= \left(
	\begin{array}{c}
		\|\mb{x}_k \|_2 \\
		0 \\
		0
	\end{array}	
	\right),\\ \label{eq_G}
	G &:= \left(
	\begin{array}{ccc}
		\sigma & 0 & \alpha\\
		\alpha c l y_- & \eta & 0  \\
		c d \epsilon l y_-(\kappa + \alpha l y y_-) & \alpha d \epsilon l^2 y y_- & \sigma + \alpha c d \epsilon l y
	\end{array}	
	\right),\\ \label{eq_Hk}
	H_k &:= \left(
	\begin{array}{ccc}
		0 & 0 & 0\\
		\alpha l y_- T \gamma_1^{k-1} & 0 & 0  \\
		(\alpha l y +2)d\epsilon l y_-^2 T \gamma_1^{k-1} & 0 & 0
	\end{array}	
	\right),
\end{align}
where
$\kappa := \|\overline{C}_k - I_{np\overline{\tau}}\|_2 = \|C - I_{np}\|_2$,
$\epsilon := \|I_{np\overline{\tau}} - \overline{C}_\infty \|_2 := \|I_{np} - C_\infty \|_2$, $\eta := \max\{1-n(\overline{\tau}+1)l,1-n(\overline{\tau}+1)s\}$. $y=\sup_k\|\mb{Y}_k\|_2$, $y_-=\sup_k\|\mb{Y}^{-1}_k\|_2$, and $c,d$ are positive constant from the equivalence of $\|\cdot\|$ and $\|\cdot\|_2$. These variables are used in the following lemmas,
\begin{lem} \label{lem_ts}
	Given the dynamics \eqref{eq_main_mat_x1}-\eqref{eq_main_mat_g1}, the following relation holds,
	\begin{align} \label{eq_lem5}
		\mb{t}_k \leq G \mb{t}_{k-1}+H_{k-1} \mb{s}_{k-1}
	\end{align}
\end{lem}
\begin{proof}
	The proof is given later in Section~\ref{sec_conv}.
\end{proof}
It should be clarified that Eq.~\eqref{eq_lem5} provides a linear iterative relation between $t_k$ and $t_{k+1}$ via
matrices, $G$ and $H_{k-1}$.  Therefore, the convergence of $t_k$ follows from spectral analysis of matrices $G$ and $H$. In other words, to prove linear convergence of $\|t_k\|_2$ toward zero, the sufficient condition is to prove $\rho(G)<1$,
as well as the linear decay of matrix $H_{k-1}$, which is straightforward from Eq.~\eqref{eq_Hk} since  $0<\gamma_1<1$.

So, we need to prove that $\rho(G)<1$ as a sufficient condition to bound $\alpha$ (the spectral radius of $G$ defined in Eq.~\eqref{eq_G} being less than $1$). This is discussed in the following lemma and proved by matrix perturbation theory.
\begin{lem} \label{lem_rhoG}
	Given the matrix $G(\alpha)$ defined in Eq.~\eqref{eq_G}, then $\rho(G(\alpha))<1$ if,
	\begin{align} \label{eq_alpharange}
		0 < \alpha < \min \{\alpha_3 , \frac{1}{n(\overline{\tau}+1)l}\}
	\end{align}
	with $ \alpha_3:= \frac{\sqrt{\delta^{2}+4 n(\overline{\tau}+1) \mu(1-\sigma)^{2} \theta}-\delta}{2 \theta}$ and
	\begin{align} \nonumber
		\delta:=& n(\overline{\tau}+1)s c d \epsilon l y_-(1-\sigma+\kappa) \\ \nonumber
		\theta :=& c d \epsilon l^{2} y y_-^2(l+n(\overline{\tau}+1)s)
	\end{align}
	
\end{lem}
\begin{proof}
	If $\alpha < \frac{1}{n(\overline{\tau}+1)l}$ then  $\eta=1-\alpha n(\overline{\tau}+1)s$, since $l \geq s$ (see details in \cite{addopt} and \cite[Chapter~3]{bubeck2015convex}). Following matrix perturbation analysis in \cite{csl2021} we set 	 $G = G_0 + \alpha \overline{G}$ with matrix $\alpha \overline{G}$ collecting the $\alpha$-dependent terms in $G$ and other independent terms in $G_0$ as,
	\begin{align}
		G_0 &=\left(\begin{array}{ccc}
			\sigma & 0 & 0 \\
			0 & \eta & 0 \\
			c d \epsilon l y_-\kappa & 0 & \sigma
		\end{array}\right) \\
		\overline{G} &=\left(\begin{array}{ccc}
			0 & 0 & 1 \\
			c l y_- & 0 & 0 \\
			c d \epsilon l y_-\left(l y y_-\right) & d \epsilon l^{2} y y_- &  c d \epsilon ly_-
		\end{array}\right)
	\end{align}	
	It is clear that for $\alpha=0$, we have $\rho(G)=\rho(G_0)=1$. This is because we know that $ 0<\sigma<1$.
	From matrix perturbation theory \cite{bhatia2007perturbation}
	and following the characteristic polynomial of $G_\alpha$ defined as
	
	\footnotesize \begin{align}\label{eq:characteristic_polynomial}
		&\left((\lambda-\sigma)^{2}-\alpha c d \epsilon l y_-(\lambda-\sigma)\right)(\lambda-1+n(\overline{\tau}+1) \alpha s)-\alpha^{3} c d \epsilon l^{3} y y_-^{2}\nonumber\\
		& \; -\alpha(\lambda-1+n(\overline{\tau}+1) \alpha s)\left(c d \epsilon l \kappa y_-+\alpha\left(c d \epsilon l^{2} y y_-^{2}\right)\right)=0.&&
	\end{align} \normalsize
	One can conclude that
	\begin{align}
		\frac{d \lambda}{d \alpha}|_{\alpha=0, \lambda=1}=-n(\overline{\tau}+1) s<0,
	\end{align}
	This implies that if we slightly increase $\alpha$ from $0$ (i.e., going from $G_0$ to $G=G_0 + \alpha \overline{G}$), the change in the eigenvalue $\lambda=1$ is towards inside the unit circle and $\rho(G_{\alpha})<1$. Next to find the range of admissible values for $\alpha$, by setting $\lambda =1$ and solving the characteristic equation~\eqref{eq:characteristic_polynomial} we get three answers: 	\begin{align}
		\alpha_{1}&=0, \quad \alpha_{2}<0, \quad \nonumber \text{and} \\
		\alpha_{3}&=\frac{\sqrt{\delta^{2}+4 n(\overline{\tau}+1) s(1-\sigma)^{2} \theta}-\delta}{2 \theta}>0.\nonumber
	\end{align}
	which implies that for $\alpha$ in the range of Eq.~\eqref{eq_alpharange} we have $\rho(G_{\alpha})<1$.
\end{proof}	
It should be noted that the bound on $\alpha$ holds for both heterogeneous delays $\tau \leq \overline{\tau}$ and homogeneous maximum delays $\tau = \overline{\tau}$. For both cases, the gradient-tracking rate $\alpha$ satisfying Eq.~\eqref{eq_alpharange} ensures the convergence.
\begin{rem}
	As a follow-up to Eq.~\eqref{eq_alpharange} in Lemma~\ref{lem_rhoG}, when the bound on time-delay $\overline{\tau}$ becomes very large, the gradient tracking rate $\alpha$ needs to become very small. This results in low convergence rate for large delays. For $\overline{\tau} \rightarrow \infty$ we have $\alpha \rightarrow 0$, which implies that the algorithm converges so slowly that it becomes difficult to implement it. Therefore, in this paper, practically we assume reasonable bounded delays and no packet-loss.  
\end{rem}
\begin{lem} \label{lem_gbar_xbar}
	The following holds for all $k>0$,
	\begin{align} \label{eq_ghk}
		\overline{\mb{g}}_k &= \mb{h}_k, \\ \label{eq_xhk}
		\overline{\mb{x}}_{k+1} &= \overline{\mb{x}}_{k} - \alpha \mb{h}_k.
	\end{align}
\end{lem}
\begin{proof}
	From Eq.~\eqref{eq_main_mat_g2}
	\begin{align} \nonumber
		\overline{\mb{g}}_k=   \frac{1}{n(\overline{\tau}+1)}  &(\mb{1}_{n(\overline{\tau}+1)} \otimes I_{p})(\mb{1}^\top_{n(\overline{\tau}+1)} \otimes I_{p})\\
		&(\overline{C}_k \widehat{\mb{g}}_{k-1} +   \nabla \mb{f}_{k} - \nabla \mb{f}_{k-1}).
	\end{align}
	Then, from column stochasticity of $\overline{C}_k$,
	\begin{align} \nonumber
		\overline{\mb{g}}_k &= \overline{\mb{g}}_{k-1} + \mb{h}_k - \mb{h}_{k-1} \\
		&= \overline{\mb{g}}_{0} + \mb{h}_k - \mb{h}_{0} =  \mb{h}_k
	\end{align}
	where the last equality follows from $\overline{\mb{g}}_{0} = \mb{h}_{0}$. Then, similar to \cite{addopt}, Eq.~\eqref{eq_xhk} follows by replacing Eq.~\eqref{eq_ghk} in Eq.~\eqref{eq_main_mat_x2}.
	%from \eqref{eq_main_mat_g}
	%\begin{align}
	%     \overline{\mb{g}}_k &=   \frac{1}{n(\overline{\tau}+1)} (\mb{1}_{n(\overline{\tau}+1)} \otimes I_{p})(\mb{1}^\top_{n(\overline{\tau}+1)} \otimes I_{p})\\
	%     &(\overline{C}_k \widehat{\mb{g}}_{k-1} +  \mb{b}^{\overline{\tau}+1}_1 \otimes (\nabla \mb{f}_{k} - \nabla \mb{f}_{k-1}))\\
	%     &= \overline{\mb{g}}_{k-1} + \mb{b}^{\overline{\tau}+1}_1 \otimes( \mb{h}_k - \mb{h}_{k-1}) \\
	%     &= \overline{\mb{g}}_{0} + \mb{b}^{\overline{\tau}+1}_1 \otimes(\mb{h}_k - \mb{h}_{0})
	%\end{align}
	%maybe we should consider $\mb{h}_k = \mb{b}^{\overline{\tau}+1}_1 \otimes \nabla \mb{f}_{k}$.
\end{proof}

\begin{lem} \label{lem_YY}
	For the proposed dynamics \eqref{eq_main_mat_x1}-\eqref{eq_main_mat_g1}, from lemma \ref{lem_nedic} we have,
	\begin{align}
		\|\mb{Y}_{k-1}^{-1} \mb{Y}_\infty -  I_{np}\|_2 \leq  y_- T \gamma_1^{k-1} \\
		\|\mb{Y}_{k}^{-1} - \mb{Y}_{k-1}^{-1}\|_2 \leq  2y_-^2 T \gamma_1^{k-1}
	\end{align}
\end{lem}
\begin{proof}
	The proof follows similar to \cite[Lemma~8]{addopt}.
\end{proof}
The following lemma provides the main results on the linear convergence of Algorithm~\ref{alg_1}.
\begin{lem} \label{lem_bubeck}
	%From lemma 9 in \cite{addopt},
	Let $s$ and $l$
	be the strong-convexity and Lipschitz-continuity constants and $ \mb{z}_+ = \mb{z} - \alpha \nabla \mb{f}(\mb{z})$ for given $\mb{z}$ and $0 < \alpha < \min \{\alpha_3 , \frac{1}{n(\overline{\tau}+1)l}\}$ with $\alpha_3$ defined as in Lemma~\ref{lem_rhoG}. Then, we have
	\begin{align} \label{eq_lem9}
		\|\mb{z}_+ - \underline{\mb{z}}^*\|_2  \leq \eta_1 \| \mb{z} - \underline{\mb{z}}^* \|_2
	\end{align}
	where $\eta_1 = \max (|1 - \alpha n l| , |1 - \alpha n s|)$.
\end{lem}
\begin{proof}
	The proof follows similar to \cite[Lemma~9]{addopt} and from \cite{bubeck2015convex}.
\end{proof}

	It should be noted that large delays may cause considerable computational overhead as the dimension of the augmented matrices scales with the time-delay bound $\overline{\tau}$. However, this trade-off is inherent to worst-case delay handling in this paper; handling delayed messages explicitly enables delay-tolerant convergence and explicit stability margins for gradient-tracking rate $\alpha$ (as shown in Eq.~\eqref{eq_alpharange}) while, on the other hand, results in higher computational costs for large delays. In this paper, considering reasonable and sufficiently small delay bounds (to avoid packet loss), the convergence analysis and computational complexity are justified.

\section{Convergence and Proof of Lemma~\ref{lem_ts}} \label{sec_conv}
%\subsection{case I}:  \eqref{eq_main_mat_x2}-\eqref{eq_main_mat_g2}
This section presents the proof of Lemma~\ref{lem_ts} and the convergence analysis in three separate steps.

\textbf{Step-I:}

First, from Eq.~\eqref{eq_main_mat_x2}, Lemma~\ref{lem_sigma} and Lemma~\ref{lem_gbar_xbar} we bound $\|\widehat{\mb{x}}_k - \mb{Y}_\infty \overline{\mb{x}}_k\| $ as,
\begin{align} \nonumber
	\|\widehat{\mb{x}}_k - \mb{Y}_\infty \overline{\mb{x}}_k\| \leq& \|\overline{C}_k \widehat{\mb{x}}_{k-1} - \mb{Y}_\infty \overline{\mb{x}}_{k-1}\| \\&+ \alpha \| \widehat{\mb{g}}_{k-1} - \mb{Y}_\infty \mb{h}_{k-1}\| \label{eq_step1_final} \\ \nonumber
	\leq& \sigma \|\widehat{\mb{x}}_{k-1} - \mb{Y}_\infty \overline{\mb{x}}_{k-1}\| \\&+ \alpha \| \widehat{\mb{g}}_{k-1} - \mb{Y}_\infty \mb{h}_{k-1}\|
\end{align}

\textbf{Step-II:}

Next we bound $\|\overline{\mb{x}}_{k} - \mb{z}^*\|_2$. From Lemma~\ref{lem_gbar_xbar},
\begin{align} \label{eq_qh}
	\overline{\mb{x}}_{k} = (\overline{\mb{x}}_{k-1} - \alpha \mb{q}_{k-1}) - \alpha (\mb{h}_{k-1} - \mb{q}_{k-1})
\end{align}
Let's define $\mb{x}_+ = \overline{\mb{x}}_{k-1} - \alpha \mb{q}_{k-1}$ as the augmented version of centralized GD step. Redefining Lemma~\ref{lem_bubeck} and Eq.~\eqref{eq_lem9} for the augmented variables, we get
\begin{align}
	\|\mb{x}_+ - {\mb{z}}^*\|_2 \leq \eta \| \widehat{\mb{x}}_{k-1} - {\mb{z}}^* \|_2
\end{align}
For the second term in \eqref{eq_qh}, from Lipschitz condition we obtain,
\footnotesize
\begin{align} 
	\|\mb{h}_{k-1} - \mb{q}_{k-1}\|_2 \leq \|\frac{1}{n(\overline{\tau}+1)} (\mb{1}_{n(\overline{\tau}+1)} \mb{1}^\top_{n(\overline{\tau}+1)}) \otimes I_{p})\|_2 l \|\widehat{\mb{z}}_{k-1} - \overline{\mb{x}}_{k-1}\|_2
\end{align} \normalsize
Then,
\begin{align}\nonumber
	\|\overline{\mb{x}}_{k} - \mb{z}^*\|_2 &\leq \|\mb{x}_+ - {\mb{z}}^*\|_2 + \alpha  l \|\mb{h}_{k-1} - {\mb{q}}_{k-1}\|_2 \\
	&\leq \eta \| \widehat{\mb{x}}_{k-1} - {\mb{z}}^* \|_2 + \alpha l  \| \widehat{\mb{z}}_{k-1} - \overline{\mb{x}}_{k-1} \|_2
	\label{eq:proofxz}
\end{align}
From Eq.~\eqref{eq_main_mat_z1} (or Eq.~\eqref{eq_main_mat_z2}) and  recalling Lemma~\ref{lem_YY} we get,
\begin{align} \nonumber
	\|\widehat{\mb{z}}_{k-1} - \overline{\mb{x}}_{k-1}\|_2   \leq& \|\mb{Y}^{-1}_{k-1}(\widehat{\mb{x}}_{k-1} - \mb{Y}_\infty \overline{\mb{x}}_{k-1})\|_2 \\\nonumber
	&+  \|\mb{Y}^{-1}_{k-1}\mb{Y}_\infty - I_{np(\overline{\tau}+1)})\overline{\mb{x}}_{k-1}\|_2 \\ \nonumber
	\leq& y_-\|\widehat{\mb{x}}_{k-1} - \mb{Y}_\infty \overline{\mb{x}}_{k-1}\|_2 \\&+ y_-T \gamma_1^{k-1} \|\widehat{\mb{x}}_{k-1}\|_2
	\label{eq:step2semifinal}
\end{align} \normalsize
where we also used the fact that  $ \|\overline{x}_{k-1}\|_2 \leq \|\widehat{x}_{k-1}\|_2 $. Then, by substituting the above in Eq.~\eqref{eq:proofxz} we get,
\begin{align} \nonumber
	\|\overline{\mb{x}}_{k} - \mb{z}^*\|_2 \leq& \alpha c l \mb{y}_- \|\widehat{\mb{x}}_{k-1} \mb{Y}_\infty \overline{\mb{x}}_{k-1} \|_2  \\
	&+ \eta  \|\overline{\mb{x}}_{k-1} - \mb{z}^*\|_2 + \alpha  l \mb{y}_- T \gamma_1^{k-1} \| \widehat{\mb{x}}_{k-1} \|_2 \label{eq_step2_final}
\end{align} \normalsize

\textbf{Step-III:}

Next, we bound $\|\widehat{\mb{g}}_{k} - \mb{Y}_\infty \mb{h}_k\|_2$. From Eq.~\eqref{eq_main_mat_g2}
\begin{align} \nonumber
	\|\widehat{\mb{g}}_{k} - \mb{Y}_\infty \mb{h}_k\|_2 &\leq \|\overline{C}_k \widehat{\mb{g}}_{k-1} - \mb{Y}_\infty \mb{h}_{k-1}\|_2 \\ \label{eq_gh}
	&+ \|(\overline{\nabla \mb{f}}_k - \overline{\nabla \mb{f}}_{k-1} - \mb{Y}_\infty( \mb{h}_{k} - \mb{h}_{k-1})\|_2
\end{align}
From Lemma~\ref{lem_sigma} and Lemma~\ref{lem_gbar_xbar},
\begin{align}
	\|\overline{C}_k \widehat{\mb{g}}_{k-1} - \mb{Y}_\infty \mb{h}_{k-1}\|_2  \leq \sigma  \|(  \widehat{\mb{g}}_{k-1} - \mb{Y}_\infty \overline{\mb{g}}_{k-1})\|_2
\end{align}
Further, the second term in \eqref{eq_gh} can be recalculated as,
\footnotesize
\begin{align} \nonumber
	&\|(\overline{\nabla \mb{f}}_k - \overline{\nabla \mb{f}}_{k-1}) - \mb{Y}_\infty( \mb{h}_{k} - \mb{h}_{k-1})\|_2 =\\ \nonumber
	&\|(I_{np(\overline{\tau}+1)} - \frac{\mb{Y}_\infty}{n}(\mb{1}_{n(\overline{\tau}+1)} \otimes I_{p})(\mb{1}^\top_{n(\overline{\tau}+1)} \otimes I_{p})) (\overline{\nabla \mb{f}}_k - \overline{\nabla \mb{f}}_{k-1}) \|_2 \\
	&\leq \epsilon l \|\widehat{\mb{z}}_k - \widehat{\mb{z}}_{k-1}\|_2
\end{align} \normalsize
which follows from the Lipschitz condition. Therefore,
\footnotesize
\begin{align} \label{eq_gh2}
	\|\widehat{\mb{g}}_{k} - \mb{Y}_\infty \mb{h}_k\|_2 &\leq \sigma \|\widehat{\mb{g}}_{k-1} - \mb{Y}_\infty \mb{h}_{k-1}\|_2 + d \epsilon l \|\widehat{\mb{z}}_k - \widehat{\mb{z}}_{k-1}\|_2
\end{align}\normalsize
To bound $\|\widehat{\mb{z}}_k - \widehat{\mb{z}}_{k-1}\|_2$ we have,
\footnotesize
\begin{align} \nonumber
	\|{\mb{h}}_{k-1} \|_2 &= \|\frac{1}{n(\overline{\tau}+1)}(\mb{1}_{n(\overline{\tau}+1)} \otimes I_{p})(\mb{1}^\top_{n(\overline{\tau}+1)} \otimes I_{p}\nabla \mb{f}(\overline{\mb{x}}_{k-1}) \|_2 \\
	&\leq l \|\overline{\mb{x}}_{k-1} - \mb{z}^*\|_2
\end{align}\normalsize
Therefore, using Eq.~\eqref{eq:step2semifinal}, we obtain,
\begin{align}  \nonumber
	\|\mb{Y}_k^{-1} \widehat{\mb{g}}_{k-1}\|_2 \leq& y_-\| \widehat{\mb{g}}_{k-1} - \mb{Y}_\infty \mb{h}_{k-1} \|_2 \\\nonumber
	&+ y_- y l \|  \overline{\mb{x}}_{k-1} - \mb{z}^* \|_2 \\\nonumber
	&+ y_-^2 y l \|  \widehat{\mb{x}}_{k-1} - \mb{Y}_\infty \overline{\mb{x}}_{k-1} \|_2 \\
	&+ y_-^2 y l T \gamma_1^{k-1}\|   \widehat{\mb{x}}_{k-1} \|_2
\end{align}
Recall that $(\overline{C}-I_{n(\overline{\tau}+1)p})\mb{Y}_\infty \overline{\mb{x}}_{k-1} = \mb{0}$. Then,
\begin{align}\nonumber
	\|\widehat{\mb{z}}_k - \widehat{\mb{z}}_{k-1}\|_2 \leq& (y_- \kappa + \alpha y_-^2 y l) \|\widehat{\mb{x}}_{k-1} - \mb{Y}_\infty \overline{\mb{x}}_{k-1}\|_2 \\\nonumber
	&+ \alpha y_-  \|\widehat{\mb{g}}_{k-1} -  \mb{Y}_\infty {\mb{h}}_{k-1}\|_2 \\\nonumber
	&+ \alpha y_- y l \|\overline{\mb{x}}_{k-1} - {\mb{z}}^*\|_2 \\
	&+ (\alpha y l +2) y_-^2 T \gamma_1^{k-1} \|\widehat{\mb{x}}_{k-1}\|_2
\end{align}
Substitute the above in Eq. \eqref{eq_gh2},
\begin{align} \nonumber
	\|{\mb{g}}_{k} - \mb{Y}_\infty \mb{h}_k\|_2 \leq& (c d \epsilon l \kappa y_- + \alpha c d \epsilon l^2 y y_-^2) \|\widehat{\mb{x}}_{k-1}   \\ \nonumber
	&- \mb{Y}_\infty \overline{\mb{x}}_{k-1}\|_2 + \alpha d \epsilon l^2 y y_-  \|\overline{\mb{x}}_k - \mb{z}^*\|_2 \\ \nonumber
	&+ (\sigma + \alpha c d \epsilon l y_-) \|\widehat{\mb{g}}_{k-1} - \mb{Y}_\infty \mb{h}_{k-1}\|_2 \\
	&+ (\alpha y l +2)d \epsilon l y_-^2 T \gamma_1^{k-1} \|\widehat{\mb{x}}_{k-1}\|_2
	\label{eq_step3_final}
\end{align}
Finally, combining Eqs. \eqref{eq_step1_final}, \eqref{eq_step2_final}, and \eqref{eq_step3_final} results in Lemma~\ref{lem_ts} and proves the convergence.

\section{Simulations} \label{sec_sim}
\subsection{Academic Example}
For the experimental simulation, we consider a quadratic cost function as Eq.~\eqref{eq_prob_lr} similar to \cite{narahari} with randomly set parameters. The number of agents is set as $n=10$ nodes.
The bound on the time-delay is set $\overline{\tau}=5$ and $\alpha=0.005$. We consider convergence over two cases of random Erdos-Renyi (ER) networks: (i) static networks where the structure of the multi-agent network is time-invariant, and (ii) dynamic (switching) networks where the network structure randomly changes every $2$ iterations. The simulations are shown in Fig.~\ref{fig_delay}. For the switching case, there exist some oscillations in the residual decay due to changes in the network topology.

\begin{figure}[]
	\centering
	\includegraphics[width=2.3in]{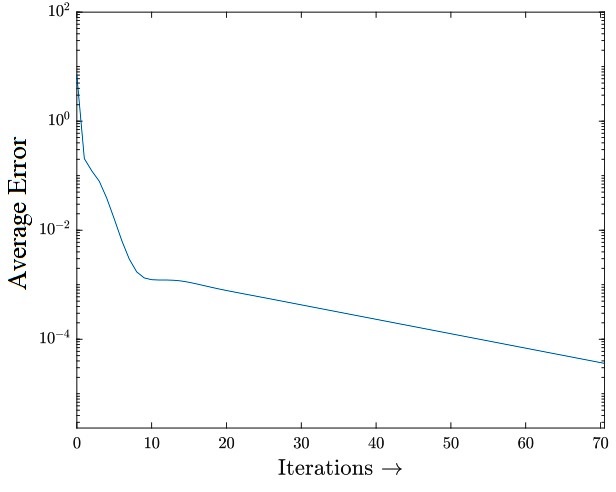} \includegraphics[width=2.3in]{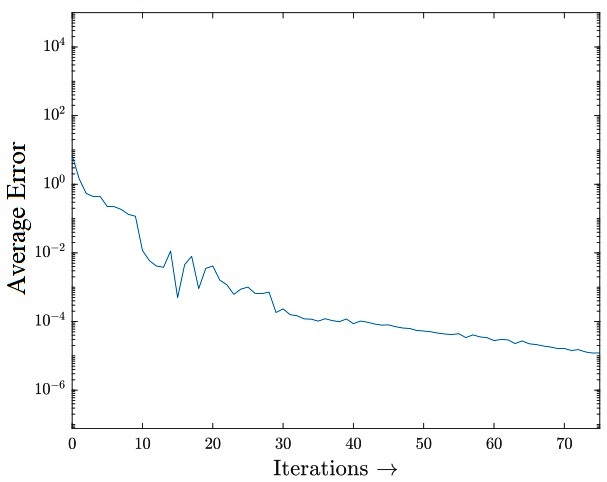}
	\caption{This figure shows the decay of the optimization residual (average error) under time-delays over (left) a static ER network and (right) a dynamic ER network. As it is clear from the figures, the algorithm converges under time-delays. There are some oscillation in the decay of the right figure, which is due to change in the network topology.} 
	\label{fig_delay}
\end{figure}

Next, we redo the simulations over an ER network to check the convergence for different values of bound on the time-delays, i.e., $\overline{\tau}= 5,10,15,20$. We set gradient-tracking rate $\alpha = 0.001$ and $\alpha = 0.005$ for this simulation. The mean-square-error (MSE) residuals at agents  for different bounds on the time-delay are shown in Fig.~\ref{fig_delays}. As it can be seen from the figure, for large value of $\overline{\tau}$, the residual decay becomes  unstable and the optimization diverges (see the residual for $\overline{\tau}= 15,20$ in the right figure for $\alpha = 0.005$). 

\begin{figure}[]
	\centering
	\includegraphics[width=2.5in]{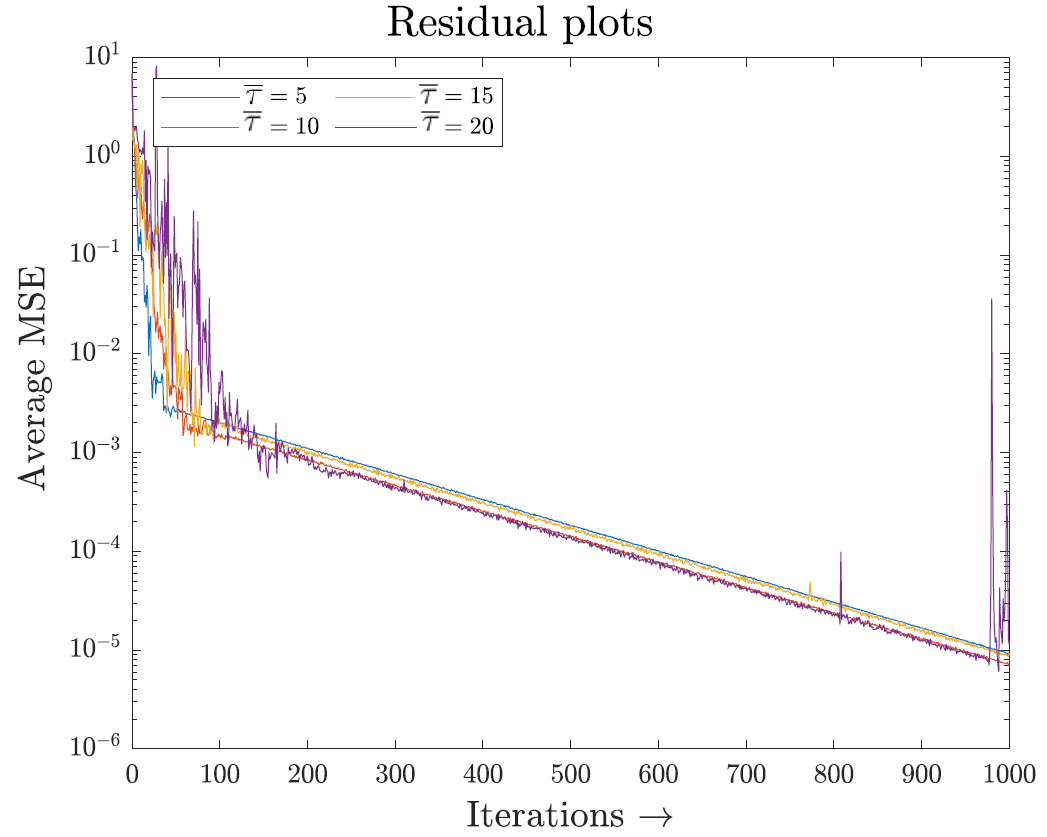}
	\includegraphics[width=2.5in]{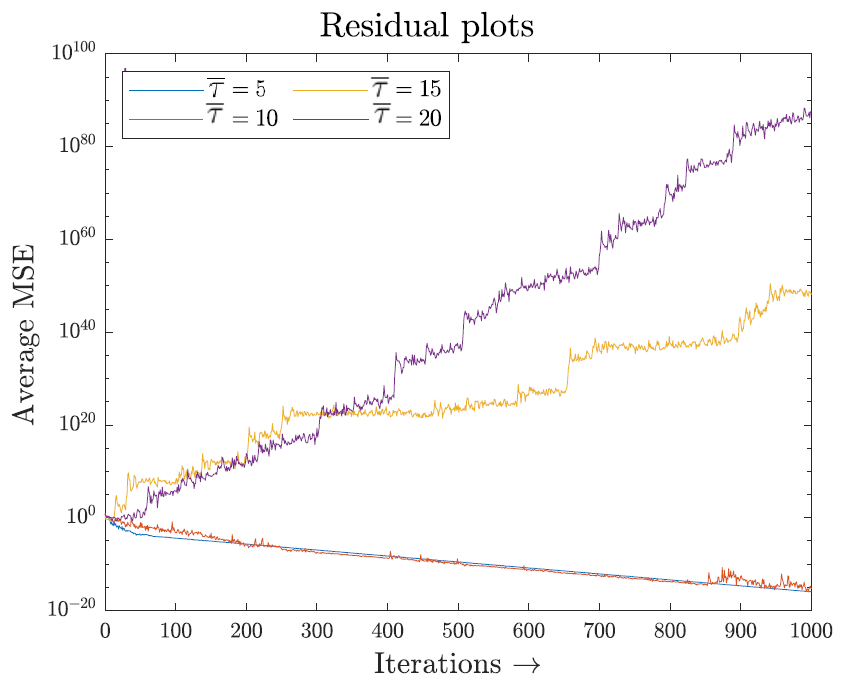}
	\caption{This figure shows the decay of the optimization residual (mean-square-error) subject to different values of time-delays over an ER network. The left figure shows the residual decay for $\alpha=0.001$ and the right figure for $\alpha=0.005$. As it is clear from the figure, for large value of $\overline{\tau}$ the residual decay becomes unstable and loses convergence.}
	\label{fig_delays}
\end{figure}
\subsection{Real Data-Set Example}
We use the MNIST dataset for distributed optimization, which is a well-known dataset in the field of machine learning and image classification. It consists of handwritten digits from $0$ to $9$ and is commonly used to train and test various classification algorithms. The dataset includes $70000$ images of handwritten digits. Each image is a $28 \times 28$ grayscale image, resulting in $784$ pixels per image, and is associated with a label from $0$ to $9$, indicating the digit it represents. A set of sampled images is shown in Fig.~\ref{fig_mnist}.
\begin{figure}[]
	\centering
	\includegraphics[width=3in]{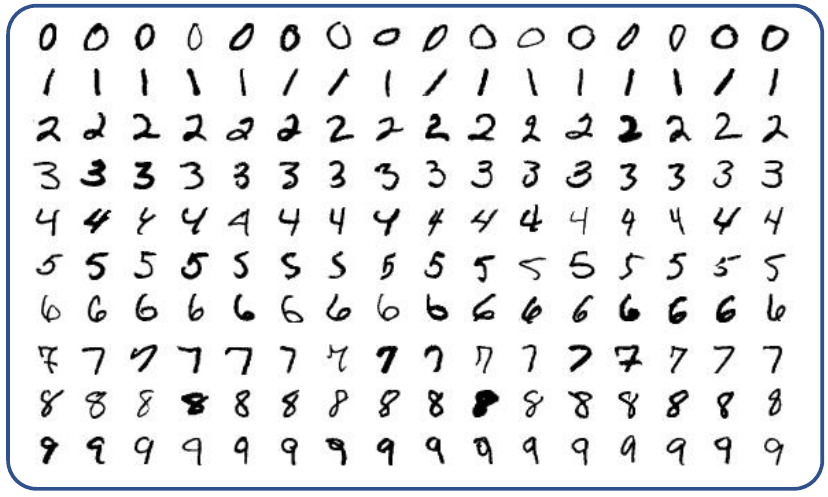}
	\caption{This figure shows a sample set of images of hand-written numbers from $0$ to $9$ taken from the MNIST data set. This data set is used for image classification via  the optimization objective \eqref{eq_F_mnist} and \eqref{eq_fij_regression}.} 
	\label{fig_mnist}
\end{figure}
The data set and image processing algorithms are taken from \cite{9827792}. We randomly select $N = 12000$ labelled images from the MNIST data set to be classified using logistic regression with a convex regularizer. The data are distributed 
among  $n=16$ agents to be cooperatively classified. In our cost optimization setup, define 
\begin{align} \label{eq_F_mnist}
	\min_{\mb{b},c} &
	F(\mb{b},c) = \frac{1}{n}\sum_{i=1}^{n} f_i(\mb{x})
\end{align}  
with every node $i$ taking a batch of $m_i=750$ sample images. Each node $i$, then, locally minimizes the following classification cost:
\begin{align}\label{eq_fij_regression}
	f_i(\mb{x}) = \frac{1}{m_i}\sum_{j=1}^{m_i} \ln(1+\exp(-(\mb{b}^\top x_{i,j}+c)y_{i,j}))+\frac{\lambda}{2}\|\mb{b}\|_2^2.
\end{align}
with $\mb{b},c$ as the classifier parameters.  The residual is defined as $F(\overline{\mb{x}}^k)-F(\mb{x}^*)$ with $\overline{\mb{x}}^k = \frac{1}{n} \sum_{i=1}^{n}\mb{x}^k_i$.
We run and compare the residual of distributed training for different existing distributed optimization techniques in the literature over an exponential network. The following optimization algorithms are used for comparison: GP \cite{nedic2014distributed}, SGP \cite{spiridonoff2020robust,nedic2016stochastic}, S-ADDOPT \cite{qureshi2020s}, and PushSAGA \cite{qureshi2021push}. The simulation results are given in Fig.~\ref{fig_mnist_sim1} for an exponential graph of $n=16$ nodes (the graph is shown in the figure). It should be mentioned that GP, SGP, S-ADDOPT, and Push-SAGA are not delay-tolerant and, thus, are simulated for delay-free case. Therefore, as expected, they show better performance in the absence of time-delays, while practically they do not converge in the presence of time-delays. On the other hand, our DTAC-ADDOPT algorithm converges in the presence of heterogeneous time-delays. For this simulation, we set $\overline{\tau}=3$.  The slower rate of convergence for DTAC-ADDOPT is due to time-delays in the data sharing as compared to the other delay-free optimization techniques.

\begin{figure}[]
	\centering
	\includegraphics[width=2.3in]{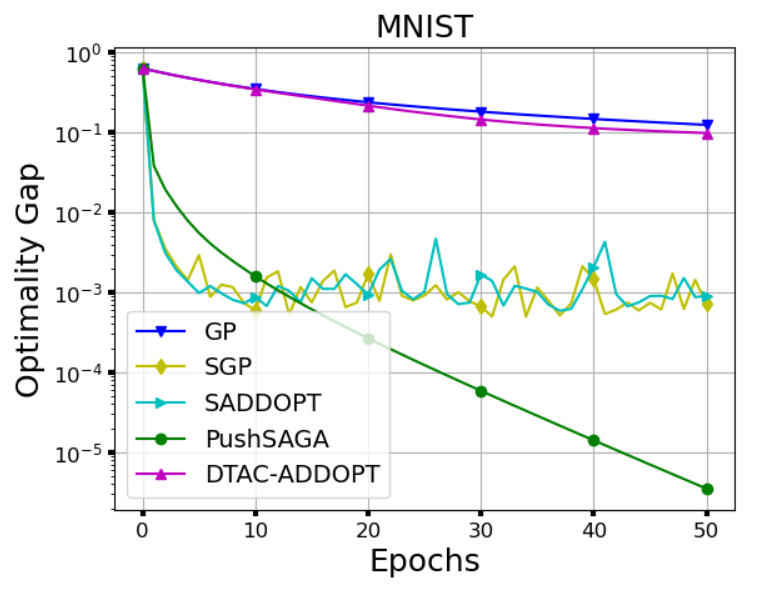} \includegraphics[width=2.1in]{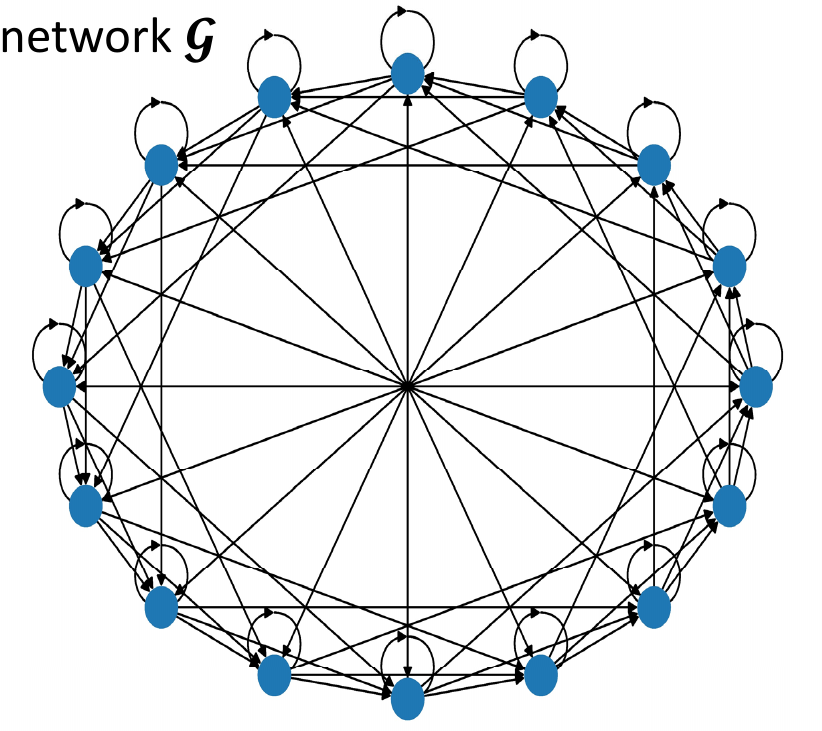}
	\caption{This simulation presents different distributed techniques over the exponential graph (given in the right figure) to optimize the objective function \eqref{eq_F_mnist} and \eqref{eq_fij_regression}. Note that only the proposed DTAC-ADDOPT is simulated subject to information-exchange delays, and the other techniques are simulated in the absence of delays.  }
	\label{fig_mnist_sim1}
\end{figure}

\section{Conclusions and Future Works} \label{sec_con}
Delay-tolerant distributed optimization over digraphs is proposed in this work. We present a distributed algorithm over a multi-agent network that is robust to time-delayed information-exchange among the agents. The delay-tolerance is shown both by mathematical proofs and experimental simulations. Future research direction includes finding a tighter bound between $\sigma_1$ and $\sigma$ based on $\overline{\tau}$ in Lemma~\ref{lem_sigma}. One can extend the convergence analysis to find maximum delay $\tau_{\max}$ for which the algorithm fails to converge when $\tau \textgreater \tau_{\max}$. Our analysis is based on bounded delays, considering no packet loss over the network, where the extension to certain classes of packet losses via standard buffering/retransmission or stochastic-analysis variants are left for future research. Moreover, distributed optimization subject to asynchronous and event-triggered operation, privacy-preserving distributed optimization \cite{mangasarian2011privacy}, and adding nonlinearities and momentum terms to reach faster convergence (similar to coupling-constrained optimization and resource allocation in \cite{ecc22}) are other open problems and directions of future research. Different applications in machine learning setups can also be considered for future research.

\bibliographystyle{elsarticle-num-names}
\bibliography{bibliography}

\end{document}